\documentclass[12pt]{amsart}
\usepackage{amssymb,mathrsfs,color,bbold,bbm}
 \usepackage[all]{xy}
\usepackage{color}

\theoremstyle{plain}
\newtheorem{thm}{Theorem}[section]
\newtheorem{prop}[thm]{Proposition}
\newtheorem{cor}[thm]{Corollary}
\newtheorem{lem}[thm]{Lemma}

\theoremstyle{definition}
\newtheorem{remark}[thm]{Remark}

\newcommand{\abs}[1]{\lvert#1\rvert}
\newcommand{\norm}[1]{\lVert#1\rVert}
\newcommand{\bigabs}[1]{\bigl\lvert#1\bigr\rvert}
\newcommand{\bignorm}[1]{\bigl\lVert#1\bigr\rVert}
\newcommand{\Bigabs}[1]{\Bigl\lvert#1\Bigr\rvert}

\newcommand{\la}{\langle}
\newcommand{\ra}{\rangle}
\newcommand{\supp}{\operatorname{supp}}
\newcommand{\N}{{\mathbb N}}

\newcommand{\E}{{\mathbb E}}
\newcommand{\R}{{\mathbb R}}
\newcommand{\X}{{\mathcal X}}
\newcommand{\C}{{\mathcal C}}
\newcommand{\lh}{{L^\Phi}}
\newcommand{\ls}{{L^\Psi}}
\newcommand{\hh}{{H^\Phi}}
\newcommand{\hs}{{H^\Psi}}

\newcommand{\om}{\omega}
\newcommand{\Om}{\Omega}
\newcommand{\Sig}{\Sigma}

\newcommand{\cQ}{{\mathcal Q}}
\newcommand{\cA}{{\mathcal A}}
\newcommand{\cT}{{\mathcal T}}
\newcommand{\cX}{{\mathcal X}}
\newcommand{\PP}{{\mathbb P}}

\newcommand{\one}{{\mathbbm{1}}}

\newcommand{\al}{\alpha}

\newcommand{\bs}{\setminus}

\newcommand{\ol}{\overline}

\def\one{\mathbb 1}

\author[N.~Gao]{Niushan Gao}
\address{Department of Mathematics and Computer Science, 
University of Lethbridge, Lethbridge, Canada T1K 3M4}
\email{gao.niushan@uleth.ca}

\author[D.~Leung]{Denny H.~Leung}
\address{Department of Mathematics, National University of Singapore, Singapore
117543}
\email{matlhh@nus.edu.sg}

\author[F.~Xanthos]{Foivos Xanthos}
\address{Department of Mathematics, 
Ryerson University, 350 Victoria St.,
Toronto, ON, M5B 2K3, Canada.}
\email{foivos@ryerson.ca}

\title[Dual Representation]{Closedness of convex sets in Orlicz spaces with
applications to dual representation of risk measures}



\thanks{The first author is a PIMS Postdoctoral Fellow. Research of the second
author is partially supported by AcRF grant R-146-000-242-114.
Part of the research for
this paper
was carried out  while
the second author was on a visit to the Department of Mathematical and
Statistical Sciences at the University of Alberta. He is grateful to his host
Vladimir Troitsky and the Department for the warm hospitality and conducive
working environment provided. The third author
acknowledges the support of an NSERC grant.}

\date{\today}

\begin{document}

\begin{abstract}
Let $(\Phi,\Psi)$ be a conjugate pair of  Orlicz functions.  A set in the Orlicz
space $L^\Phi$ is said to be order closed if 
it is closed with respect to dominated convergence of sequences of functions.
 A well known  problem arising from the theory of risk measures in financial
mathematics asks whether
order closedness of a convex set in $L^\Phi$ characterizes closedness with
respect to the topology $\sigma(L^\Phi,L^\Psi)$.  (See \cite[p.\
3585]{Owari:14}.)
In this paper, we show that for a norm bounded convex set in $L^\Phi$, order
closedness and $\sigma(L^\Phi,L^\Psi)$-closedness are indeed equivalent.
In general, however, coincidence of order closedness and
$\sigma(L^\Phi,L^\Psi)$-closedness 
of convex sets in $L^\Phi$ is equivalent to the validity of the Krein-Smulian
Theorem for the topology $\sigma(L^\Phi,L^\Psi)$; that is, a convex set is
$\sigma(L^\Phi,L^\Psi)$-closed if and only if it is closed with respect to the
bounded-$\sigma(L^\Phi,L^\Psi)$ topology.
As a result, we show that order closedness and
$\sigma(L^\Phi,L^\Psi)$-closedness of convex sets in $L^\Phi$ are equivalent  if
and only if either $\Phi$ or $\Psi$ satisfies the $\Delta_2$-condition.
Using this, we prove the surprising result that:  \emph{If (and only if) $\Phi$ and $\Psi$ both fail the $\Delta_2$-condition, then there exists a  
coherent risk measure on $\lh$ that has the Fatou property but fails the
Fenchel-Moreau dual representation with respect to the dual pair $(L^\Phi, L^\Psi)$}.
A similar analysis is carried out for the dual pair of Orlicz hearts
$(H^\Phi,H^\Psi)$.
\end{abstract}

\maketitle

\section{Introduction}

In the seminal paper \cite{ADEH:99}, a theoretical foundation was laid for the
problem of quantifying the risk 
of a financial position in terms of coherent risk measures.  The theory of risk
measures has since been an active and fruitful area of research in Mathematical
Finance (cf.\
\cite{Arai:10,Arai:14,BF:10,Biagini:11,CL:09,Delbaen:09,FKM:13,Fri:02,
Jouini06,
Orihuela:12}).  It has also motivated new developments in Convex Analysis and
Functional Analysis (cf.\ \cite{DRAPE16,FILL09,GUO10,Owari:14}).

In the axiomatic treatment of risk measures, financial  positions are modeled by
a vector space $\cX$, which includes constants, of random variables on a
probability space $(\Om,\Sig,\PP)$.
A {\em coherent risk measure} on $\cX$ is a functional $\rho:\cX\to
(-\infty,\infty]$ that is proper (i.e., not identically $\infty$) and satisfies
the following properties:
\begin{enumerate}
\item (Subadditive) $\rho(X_1 + X_2) \leq  \rho(X_1) +\rho(X_2)$ for all
$X_1,X_2\in \cX$.
\item (Monotone) $\rho(X_1) \leq \rho(X_2)$ if $X_1,X_2\in \cX$ and $X_1 \geq
X_2$ a.s.
\item (Cash additive) $\rho(X+m\one) = \rho(X) - m$ for any $X\in \cX$ and any
$m\in \R$.
\item (Positively homogeneous) $\rho(\lambda X) = \lambda \rho(X)$ for any $X
\in
\cX$ and any $0 < \lambda \in \R$.
\end{enumerate}
Observe that a coherent risk measure is convex, i.e.,
\[ \rho(\lambda X_1 + (1-\lambda)X_2) \leq  \lambda\rho(X_1)
+(1-\lambda)\rho(X_2) \]
\text{ for all $X_1,X_2\in \cX$ and all $\lambda\in [0,1]$}.

An important topic in the theory of risk measures is to determine when a risk
measure on $\cX$ admits a  representation with respect to some duality involving
$\cX$.
The first major result in this direction was obtained by Delbaen [10], who used
as model space $\cX$ the space of all bounded random variables $L^\infty(\PP)$
and considered the duality $(L^\infty,L^1)$.
\begin{thm}[Delbaen]\label{del}
The following are equivalent for every coherent risk measure $\rho$ on
$L^\infty(\PP)$.
\begin{enumerate}
\item There is a set $\cQ$ of  nonnegative random variables with expectation $1$
such that
\[ \rho(X) = \sup_{Y\in \cQ}\E[-XY] \text{ for any $X\in L^\infty$},\]
\item $\rho$ satisfies the {\em ($L^\infty$-) Fatou property}:
\[ \rho(X) \leq \liminf_n \rho(X_n) \text{ whenever } (X_n) \text{ is {\em
bounded} in $L^\infty$ and $X_n \xrightarrow{a.s.} X$.}\]
\end{enumerate}
\end{thm}
In Delbaen's theorem, the set $\cQ$ can be interpreted as a set of probability
distributions (scenarios) and the risk measure of $X$ is obtained as the worst
expected loss over the set of scenarios (stress tests).  In general, such dual
representations play an important role in optimization problems and portfolio
selection.  

The representation in (1) of Theorem \ref{del} is  connected
with $\sigma(L^\infty,L^1)$-lower semicontinuity of $\rho$ via the
Fenchel-Moreau Duality Theorem in convex analysis.  Here,
$\sigma(L^\infty,L^1)$-lower semicontinuity of $\rho$ refers to the property
that the sets 
\[\{\rho \leq \lambda\} = \{X\in L^\infty(\PP): \rho(X) \leq \lambda\}\] are
$\sigma(L^\infty,L^1)$-closed for any $\lambda \in \R$.
On the other hand, condition (2) in Theorem \ref{del} is equivalent to the fact
that the sets $\{\rho\leq \lambda\}$ are closed with respect to dominated
convergence of sequences.

In Theorem \ref{del}, as in other early framework for risk measures, the model
space consists of bounded financial positions.  The reader may refer to
F$\ddot{\text{o}}$llmer and Schied [17] for a comprehensive treatment of the
main results in this setting.
More realistic models of financial positions may involve unbounded random
variables; this motivates the study of  risk measures on model spaces beyond
$L^\infty$.  The Orlicz spaces $L^\Phi$ and Orlicz hearts $H^\Phi$, being
natural classes of Banach function spaces that generalize the $L^p$ spaces, have
emerged as important model spaces of (unbounded) financial positions.
Contributions to the study of risk measures on $L^\Phi$ and $H^\Phi$ may be
found in Cheridito and Li [8], Biagini and Fritelli [5] and more recently in
Kr$\ddot{\text{a}}$tschmer, Schied and Z$\ddot{\text{a}}$hle [24] and Gao and
Xanthos [19].  See also Gao et al [18] and the references therein.

Let $(\Phi,\Psi)$ be a conjugate pair of Orlicz functions.  (Refer to \S 2 for
definitions and basic facts concerning Orlicz spaces.)
In the paper [5], Biagini and Fritelli initiated the study of representation of
risk measures on $L^\Phi$ with respect to the duality  $(L^\Phi,L^\Psi)$.
Observing that dominated convergence of a sequence of random variables in
$L^\Phi$ implies $\sigma(L^\Phi,L^\Psi)$-convergence, they proposed the
following version of the Fatou property on $L^\Phi$:
A functional $\rho$ on $L^\Phi$ is said to satisfy the {\em Fatou property} if 
\[ \rho(X) \leq \liminf_n \rho(X_n) \text{ whenever } (X_n) \text{ is {\em order
 bounded} in $L^\Phi$ and $X_n \xrightarrow{a.s.} X$.}\]
They then claimed a result similar to Theorem \ref{del}  for every conjugate
pair $(L^\Phi,L^\Psi)$.
Specifically,
it was claimed that 
\begin{align*}
&\text{for every coherent risk measure } \rho:L^\Phi\to (-\infty,\infty],\\
&
\text{$\rho$ has the Fatou property if and only if} \\
(*) \quad \quad& \text{there is a set $\cQ$ of nonnegative random variables in
$L^\Psi$,}\\\
& \text{each having expectation $1$, such that}\\
&\rho(X) = \sup_{Y\in \cQ}\E[-XY] \text{ for any $X\in L^\Phi$}.\end{align*}
As in Theorem \ref{del}, validity of ($*$) is closely linked to the equivalence
of $\sigma(L^\Phi,L^\Psi)$-closedness and closedness under dominated convergence
for convex sets.
Their proof is based on an assertion that every Orlicz space enjoys a technical
property (with regard to convex sets) which they called the  $C$-property. 
Unfortunately, the validity of
the $C$-property has been disproved in [19] when $\Psi$ satisfies the
$\Delta_2$-condition.
Thus, the verity of the foregoing equivalence ($*$) for every conjugate pair
$(L^\Phi,L^\Psi)$ remained an important outstanding problem in the theory of
dual representation of risk measures.
See, e.g., \cite[p.~3585]{Owari:14}, where this problem is raised explicitly.

The main result of this paper is a comprehensive solution to this problem. It is
shown that the equivalence ($*$) holds for a conjugate pair $(L^\Phi,L^\Psi)$ if
and only if 
either $\Phi$ or $\Psi$ satisfies the $\Delta_2$-condition.
First we relate the problem to the equivalence of order closedness and
$\sigma(L^\Phi,L^\Psi)$-closedness for convex sets in $L^\Phi$.  It is shown
that for a {\em norm bounded} convex set in $L^\Phi$, order closedness and
$\sigma(L^\Phi,L^\Psi)$-closedness are indeed equivalent.  
As a result, the validity of ($*$) for a dual pair $(L^\Phi,L^\Psi)$ is
determined by whether the  Krein-Smulian Theorem for the topology
$\sigma(L^\Phi,L^\Psi)$ holds.  We complete the proof of the main result by
showing that the Krein-Smulian Theorem for the topology $\sigma(L^\Phi,L^\Psi)$
holds if and only if either $\Phi$ or $\Psi$ satisfies the $\Delta_2$-condition.

In the last section, we also investigate the dual representation problem on
$H^\Phi$ with respect to the dual pair $(H^\Phi,H^\Psi)$.  This complements the
results for the dual pair $(L^\Phi,L^\Psi)$  in Section~\ref{sec2}, for the dual
pair $(L^\Phi,H^\Psi)$  by Gao and Xanthos
[19] and  for the dual pair $(H^\Phi, L^\Psi)$ by Cheridito and Li [8].

\section{Orlicz spaces}

In this section, we collect the basic facts regarding Orlicz spaces and set the 
notation.
We adopt  \cite{AB:06} and  \cite{ES:02,RR:91} as standard references for
unexplained terminology
and facts
on Banach lattices and Orlicz spaces, respectively. 
Recall that a function $\Phi:[0,\infty) \rightarrow[0,\infty)$ is called an
\emph{Orlicz function} if it is convex, increasing, and $\Phi(0)=0$. 
Define the
\emph{conjugate function} of $\Phi$ by $$\Psi(s)=\sup\{ts-\Phi(t) : t \geq 0\}$$
for all $s\geq 0$. 
If $\lim_{t\to\infty}\frac{\Phi(t)}{t}=\infty$ (equivalently, if $\Psi$ is
finite-valued), then $\Psi$ is also an Orlicz function, and its conjugate is
$\Phi$.
{Throughout this paper, $(\Phi,\Psi)$ stands
for an
Orlicz pair such that $\lim_{t\to\infty}\frac{\Phi(t)}{t}=\infty$ and that
$\Phi(t)>0$ for any $t>0$}. The restrictions on
$\Phi$ are minor as they only eliminate the case where $L^\Phi$ coincides with
$L^1$ or $L^\infty$, in which cases our main results are either trivial or
known.

Throughout this paper, $(\Omega,\Sigma,\mathbb{P})$ stands for a nonatomic
probability space.
The \emph{Orlicz
space} $L^\Phi:=L^\Phi(\Omega,\Sigma,\mathbb{P})$ is the space of all
real-valued random variables
$X$ (modulo a.s.~equality) such that
$$\norm{X}_\Phi:=\inf\left\{\lambda>0:\E\left[\Phi\left(\frac{|X|}{\lambda}
\right)
\right]\leq 1\right\}<\infty.$$
The norm $||\cdot||_{\Phi}$
on $L^\Phi$ is called the \emph{Luxemburg norm}.  The subspace of
$L^\Phi$ consisting
of all $X\in L^\Phi $ such that
$$
\E\left[\Phi\left(\frac{|X|}{\lambda}\right)\right]<\infty\;\;\mbox{ for all }
\lambda>0$$
is conventionally called the \emph{Orlicz heart} of $L^\Phi$ and is denoted by
$H^\Phi$. It is well known that $L^\infty\subset H^\Phi\subset L^\Phi\subset
L^1$ and that $H^\Phi$ is a norm closed subspace of $L^\Phi$.
We always endow the conjugate Orlicz space $L^\Psi$ and
the conjugate Orlicz heart $H^\Psi$ with
the \emph{Orlicz norm}
\[
\norm{Y}_\Psi := \sup_{X\in L^\Phi, \,\norm{X}_\Phi\leq1}\bigabs{\mathbb{E}[XY]}
\]
for all $Y\in L^\Psi$, which is equivalent to the Luxemburg norm on $L^\Psi$.

An Orlicz function $\Phi$ satisfies the $\Delta_2$-{\em condition} if 
there exist $t_0\in(0,\infty)$ and $k\in\mathbb{R}$
such that
$\Phi(2t)<k\Phi(t)$ for all $t\geq t_0$.
It is well known that $\ls=(\hh)^*$ and that $\ls $,  being the order continuous dual of $\lh$, is a lattice  ideal in
$(\lh)^*$.
Moreover, the following conditions are equivalent.
\begin{enumerate}
\item $\ls=(\lh)^*$, 
\item $L^\Phi=H^\Phi$, 
\item The Orlicz function $\Phi$ satisfies the
$\Delta_2$-condition,
\item $\lh$ \emph{has order continuous norm},
i.e., $$X_\alpha\downarrow,\;\inf_\alpha X_\alpha=0\mbox{ in }\lh\quad
\Longrightarrow
\quad\inf_\alpha\norm{X_\alpha}_\Phi=0.$$
\end{enumerate}
A sequence $(X_n)$ in $L^\Phi$ is {\em order bounded} if there exists $X\in
L^\Phi$ such that $|X_n| \leq X$ a.s.~for all $n$.
A sequence $(X_n)$ in $L^\Phi$  \emph{order converges} to
$X\in L^\Phi$,
written as
$X_n\xrightarrow{o}X$ in $L^\Phi$, if $X_n\xrightarrow{a.s.}X$ and $(X_n)$ is
order bounded in $\lh$.
If $L^\Phi$ has order continuous norm, then 
$$X_n\xrightarrow{o}X\mbox{ in }L^\Phi\quad \Longrightarrow\quad
\norm{X_n-X}_\Phi\rightarrow0.$$

\section{Dual representation with respect to the pair
$(L^\Phi,L^\Psi)$}\label{sec2}

This section is the main part of the paper, where we show that the equivalence
($*$) in the Introduction holds if and only if either $\Phi$ or $\Psi$ satisfies
the $\Delta_2$-condition.  

For a subset $\C$ of $\lh$, define its \emph{order
closure} in $\lh$ to be the set 
$$\overline{\C}^{o}:=\big\{X\in \lh:X_n\xrightarrow{o}X \mbox{ for some sequence
}(X_n)\mbox{ in }\C \big\}.\footnote{In the definition of order closures, we can equivalently use nets instead of sequences. Indeed, since $L^\Phi $ has the countable sup property, if $X_\alpha\xrightarrow{o}X$ in $L^\Phi$, then there exists countably many $(\alpha_n) $ such that $X_{\alpha_n}\xrightarrow{o}X$ in $L^\Phi$.}$$
We say that $\C$ is \emph{order closed} in $\lh$ if
$\C=\overline{\C}^{o}$. 
In spite of the terminology, 
\emph{$\overline{\C}^{o}$ is not necessarily order closed.}
By Dominated Convergence Theorem, 
$$X_n\xrightarrow{o}X\mbox{ in }\lh\quad\Longrightarrow\quad
\E[X_nY]\rightarrow\E[XY],\; \mbox{ for any }Y\in
\ls.$$
Thus $$\overline{\C}^{o}\subset \ol{\C}^{\sigma_s(\lh,\ls)} \subset
\overline{\C}^{\sigma(\lh,\ls)}$$ for any set $\C\subset L^\Phi$, where $
\ol{\C}^{\sigma_s(\lh,\ls)}$ denotes the $\sigma(\lh,\ls)$-{\em sequential}
closure of $\C$.
In particular,
every $\sigma(\lh,\ls)$-closed set is order closed.

We begin by examining the connection between  ($*$) and the
equivalence of order closedness and $\sigma(L^\Phi,L^\Psi)$-closedness of convex
sets in $L^\Phi$.
The next two propositions are essentially known.  We include the proofs for the
sake of completeness.

\begin{prop}\label{prop3.1}
Let  $\rho:L^\Phi \to (-\infty,\infty]$ be a proper convex functional. Then
$\rho$ has the Fatou property if and only if the set $\C_\lambda = \{X\in
L^\Phi: \rho(X) \leq \lambda\}$ is order closed for any $\lambda \in \R$.
\end{prop}

\begin{proof}
The fact that $\C_\lambda$ is order closed if $\rho$ satisfies the Fatou
property is immediate from the definitions.
Conversely, suppose that $\C_\lambda$ is order closed for any $\lambda\in \R$.
Let $(X_n)$  be a sequence in $L^\Phi$ that order converges  to $X\in L^\Phi$.
If $\liminf_n\rho(X_n) =\infty$, then $\rho(X) \leq \liminf_n\rho(X_n)$
trivially.
Otherwise, let $\lambda \in \R$ be such that $\lambda > \liminf_n\rho(X_n)$.
Choose a subsequence $(X_{n_k})$ so that $\rho(X_{n_k}) < \lambda$ for all $k$.
Then $X_{n_k} \in \C_\lambda$ for all $k$ and $X_{n_k}\stackrel{o}{\to}X$.  Thus
$X\in \C_\lambda$ and $\rho(X)\leq \lambda$.
As this applies to any $\lambda> \liminf_n\rho(X_n)$, $\rho(X) \leq
\liminf_n\rho(X_n)$.
Therefore, $\rho$ has the Fatou property.
\end{proof}

A functional $\rho:L^\Phi\to (-\infty,\infty]$ is said to be
$\sigma(\lh,\ls)$-{\em lower semicontinuous} if the set $\{X\in L^\Phi: \rho(X)
\leq \lambda\}$ is $\sigma(\lh,\ls)$-closed for all $\lambda\in \R$.

\begin{prop}\label{prop3.2}
The following are equivalent for a proper convex functional $\rho:L^\Phi \to
(-\infty,\infty]$ .
\begin{enumerate}
\item $\rho$ is lower $\sigma(\lh,\ls)$-semicontinuous.
\item Define $\rho^*(Y)=\sup_{X\in L^\Phi}(\E[XY] -\rho(X))$ for any $Y\in
L^\Psi$.  Then $\rho(X)=\sup_{Y\in L^\Psi}(\E[XY]-\rho^*(Y))$ for any $X\in
L^\Phi$.
\end{enumerate}
If $\rho$ is a coherent risk measure, then the conditions above are also
equivalent to 
\begin{enumerate}
\item[(3)] There is a set $\cQ$ of nonnegative random variables in $L^\Psi$,
{each having expectation $1$, such that}
$\rho(X) = \sup_{Y\in \cQ}\E[-XY]$ \text{ for any $X\in L^\Phi$}.
\end{enumerate}
\end{prop}

\begin{proof}

\noindent(1) $\iff$ (2) is a consequence of the  Fenchel-Moreau Duality Theorem
(\cite[Theorem~1.11]{B:11}).

For the rest of the proof, assume that $\rho$ is a coherent risk measure on
$L^\Phi$.
The implication (3) $\implies$ (1) is trivial.

\noindent(2) $\implies$(3).
Assume that (2) holds.  Choose $Z\in L^\Phi$ so that $\rho(Z) < \infty$. From
positive homogeneity of $\rho$, one deduces easily that $\rho^*(\lambda W) =
\lambda\rho^*(W)$ for any $W\in L^\Psi$ and any $0 < \lambda \in \R$.
Hence $\rho^*(W) = 0$ or $\infty$ for any $W \in L^\Psi$.
Suppose that $W \in L^\Psi$ and $\rho^*(W) =0$.
Let $A  = \{\om: W(\om) >0\}$ and set $X = \one_A$.  Then $0 \leq X\in L^\Phi$
and $\E[XW] \geq 0$. 
For any
 $\lambda>0$, by monotonicity of $\rho$,
\[ \lambda
\E[XW] +\E[ZW]=\E[ (\lambda
X+Z)W]\leq  \rho(\lambda X+Z)+\rho^*(W)\leq
\rho(Z)+\rho^*(W)<\infty.\]
Since $\lambda$ is arbitrary, $\E[W\one_{\{W > 0\}}] =\E[XW] = 0$ and hence $W
\leq 0$ a.s.
Define 
\[\mathcal{Q}:=\{Y\in L^\Psi:\rho^*(-Y)=0\}.\]
 By the preceding argument, $Y \geq 0$ a.s.\ for any $Y \in \cQ$.
 Also, by (3) and the fact that $\rho^*(-W) = \infty$ for all $W\in L^\Psi \bs
\cQ$, we see that $\rho(X) = \sup_{Y\in\mathcal{Q}}\E[-XY]$ for any $X\in
L^\Phi$.
Finally, for any $m\in\R$, 
\[\rho(Z)-m=\rho(Z+m\one)=\sup_{Y\in\mathcal{Q}}(-\E[
ZY]-m\E[Y]).\] 
Whence, for any $Y \in \cQ$ and any $m\in \R$,
\[\rho(Z)\geq 
-\E[ZY]+m(1-\E[Y]).\]
Therefore,  $\E[Y] = 1$ for any $Y\in \cQ$.
This completes the proof of (2) $\implies$ (3).  
\end{proof}

Let $\rho$ be a coherent risk measure on $L^\Phi$.  
The equivalence ($*$) asserts that any coherent risk measure on $L^\Phi$ has the
Fatou
property if and only if it satisfies condition (3) of Proposition \ref{prop3.2}.
Since the set $\C = \{X\in
L^\Phi: \rho(X) \leq 0\}$ is convex, by the preceding propositions, the
equivalence 
($*$) holds for $L^\Phi$ if for any
convex set in $L^\Phi$, order closedness and  $\sigma(L^\Phi,L^\Psi)$-closedness
are equivalent.
With this in mind, the following property was introduced in \cite{BF:10}.
The topology $\sigma(L^\Phi,L^\Psi)$ is said to have the {\em $C$-property} if 
for given any net $(X_\al)$ in $L^\Phi$ that $\sigma(L^\Phi,L^\Psi)$-converges
to $X\in L^\Phi$, there is a sequence $(Z_n)$ of convex combinations of
$(X_\al)$ so that $Z_n \stackrel{o}{\to} X$.
Evidently, if $\sigma(L^\Phi,L^\Psi)$ has the $C$-property, then $
\overline{\C}^{\sigma(\lh,\ls)}\subset \overline{\C}^{o}$
for any convex set $\C$ in $L^\Phi$.  Since the reverse inclusion always holds,
it follows that order closedness and $\sigma(\lh,\ls)$-closedness would coincide
for
any convex set in $\lh$.
Consequently, the equivalence ($*$) would hold for $L^\Phi$.
Unfortunately, the next proposition shows that the $C$-property occurs rather
sparsely.

\begin{prop}\label{cpp}
$\overline{\C}^{o}=\overline{\C}^{\sigma(\lh,\ls)}$  for
every convex
set
$\C$ in $\lh$ if and only if $\Phi$ satisfies the $\Delta_2$-condition.
\end{prop}

\begin{proof}
If $\Phi$ satisfies the
$\Delta_2$-condition, then $\lh = H^\Phi$ and thus
$\sigma(\lh,\ls)$ is the weak topology on $L^\Phi$.  By Mazur's Theorem,
$\overline{\C}^{\sigma(\lh,\ls)}=\overline{\C}^{\norm{\cdot}}$.
Since every norm
convergent sequence admits a subsequence that order converges to the same
limit (see, e.g., \cite[Lemma~3.11]{GX:14}),
$\overline{\C}^{\sigma(\lh,\ls)}\subset \overline{\C}^{o}$.
Therefore, $\overline{\C}^{o}=\overline{\C}^{\sigma(\lh,\ls)}$.

Conversely, suppose that $\Phi$ fails the $\Delta_2$-condition. 
By
\cite[pp.~139, Theorem~5]{RR:91} (or \cite[Theorem~4.51]{AB:06}), there exist a
sequence $(X_n)$ of pairwise disjoint random variables in $\lh$ and a closed
sublattice $\X$ of $\lh$
such that the map $\mathcal{T}:\ell^\infty\rightarrow \mathcal{X}$ defined by
$\cT\big((a_n)_n\big)=\sum_na_nX_n$ (pointwise sum) is a Banach lattice
isomorphism.
Denote by $e$ the identically $1$ sequence in $\ell^\infty$. If  $Y\in L^\Psi$,
then 
$$\sum_n\bigabs{\E[X_nY]}\leq
\sum_n\E[X_n\abs{Y}]=\E\big[\cT e \abs{Y}\big]<\infty.$$
Hence $(\E[X_nY]) \in \ell^1$.
Thus, if $w = (a_n) \in \ell^\infty$, then
\begin{equation}\label{continuity} \E[(\cT w)Y] = \E\Big[\big(\sum
a_nX_n\big)Y\Big] = \sum a_n
\E[X_nY] = \big\la (\E[X_nY]), w\big\ra.
\end{equation}

By Ostrovskii's Theorem (cf.~\cite[Theorem~2.34]{HZ:07}), there exist a subspace
$\mathcal{W}$ of $\ell^\infty$ and
$w\in \overline{\mathcal{W}}^{\sigma(\ell^\infty,\ell^1)}$ such that $w$ is not
the
$\sigma(\ell^\infty,\ell^1)$-limit of any sequence in $\mathcal{W}$. 
Let $\C=\mathcal{T}(\mathcal{W})$.  Obviously, $\C$ is a convex set in $L^\Phi$.
Take a net  $(w_\al)\subset \mathcal{W}$ that
$\sigma(\ell^\infty,\ell^1)$-converges
to $w\in \ell^\infty$.
Let $Y\in L^\Psi$.  By \eqref{continuity},
\[ \E[(\cT w_\al)Y] = \la (\E[X_nY]), w_\al \ra \to \la (\E[X_nY]), w\ra
=\E[(\cT w)Y].\] 
Thus $(\cT w_\al)$ $\sigma(\lh,\ls)$-converges to $\cT w$ and $\cT
w\in\overline{\C}^{\sigma(\lh,\ls)}$.
Suppose, if possible, that $\mathcal{T}w\in \overline{\C}^{o}$.
Take a sequence $(w_k)$ in $\mathcal{W}$ such that
$\mathcal{T}w_k\xrightarrow{o}\mathcal{T}w$ in $\lh$.
Clearly, $(\mathcal{T}w_k)$, being order bounded, is norm bounded in $\lh$, so
that $(w_k)$ is norm bounded in $\ell^\infty$.
Write $w_k=(a_n^k)_n$ and $w=(a_n)_n$. Since the $X_n$'s are disjoint and 
$\mathcal{T}w_k=\sum_na_n^kX_n\xrightarrow{a.s.}Tw=\sum_na_nX_n$, 
$\lim_k a_n^k=a_n$ for each $n$. 
It follows that $(w_k)$ ${\sigma(\ell^\infty,\ell^1)}$-converges to $w$, 
contrary to the choice of $w$.
\end{proof}

However, the equality $\overline{\C}^o=\overline{\C}^{\sigma(L^\Phi,L^\Psi)}$
does hold in general for {\em norm bounded} convex sets  $\C \subset L^\Phi$.

\begin{thm}\label{bdd-C}
Let $\C$ be any norm bounded convex set in $L^\Phi$. Then
$\overline{\C}^o=\overline{\C}^{\sigma(L^\Phi,L^\Psi)}$. In particular,
$\overline{\C}^o$ is order closed, and $\C$ is
order closed if and only if it is $\sigma(L^\Phi,L^\Psi)$-closed. \end{thm}

\begin{proof}
As has been observed, the inclusion $\overline{\C}^{o}\subset
\overline{\C}^{\sigma(\lh,\ls)}$ always holds.  To prove the reverse inclusion,
it suffices to show that if $\C$ is a convex subset of the unit ball of $L^\Phi$
such that  $0 \in \ol{\C}^{\sigma(\lh,\ls)}$, then $0 \in \ol{\C}^o$.
We divide the proof into three steps.

\noindent \underline{Step I}: For each $n\geq 1$ and each $0\leq Y\in L^\Psi$, 
there exists a pair of disjoint random variables $Z_{Y,n}$ and $W_{Y,n}$
satisfying the following:
\begin{enumerate}
\item $Z_{Y,n}\in L^\Phi$ and $W_{Y,n}\in H^\Phi$,
\item $Z_{Y,n} + W_{Y,n} \in \C$,
\item $\E[\Phi(\abs{Z_{Y,n}})] \leq \frac{1}{2^n}$,
\item $\E[\abs{W_{Y,n}}Y] \leq 1$.
\end{enumerate}

Since $L^\Psi$ is a lattice ideal in $(L^\Phi)^*$,  by
\cite[Theorem~3.50]{AB:06}, the topological dual of $L^\Phi$ under
$\abs{\sigma}(L^\Phi,L^\Psi)$ is precisely $L^\Psi$. Thus, by Mazur's Theorem,
$$0\in
\overline{\C}^{\sigma(L^\Phi,L^\Psi)}=\overline{\C}^{\abs{\sigma}(L^\Phi,L^\Psi)
}
,$$so that there exists $X\in\C$ such that $\E[\abs{X}Y]\leq 1$.
Since $\C$ is in the unit ball, we have that $\E[\Phi(\abs{X})] \leq 1$.
Now take $k\geq 1$ such that $$\E\big[\one_{\{\abs{X} >
k\}}\Phi(\abs{X})\big]\leq \frac{1}{2^n}.$$
Set $$Z_{Y,n} = X\one_{\{\abs{X} > k\}}\;\;\mbox{ and }\;\;W_{Y,n} =
X\one_{\{\abs{X}
\leq 
k\}}.$$
Clearly, $Z_{Y,n}$ and $W_{Y,n}$ are disjoint.
Conditions (1)-(3) are easily verified.
Condition (4) holds because $$\E[\abs{W_{Y,n}}Y]\leq \E[\abs{X}Y]\leq 1.$$

\noindent \underline{Step II}: There exist sequences $(Z_n)$ and $(W_n)$ in
$L^\Phi$ such
that for each $n\geq 1$,
\begin{enumerate}
\item $X_n := Z_n + W_n\in \C$,
\item $\E[\Phi(|Z_n|)]\leq \frac{1}{2^n}$,
\item $\norm{W_n}_\Phi \leq \frac{1}{2^n}$.
\end{enumerate}

Keep the notation of Step I.  For any $n \geq 1$, define
$\mathcal{A}_{n}=\{W_{Y,n}: 0\leq Y\in L^\Psi\}\subset H^\Phi$.
Let $Y_1,\dots, Y_k\in L^\Psi$ and let $\varepsilon >0$ be given.
Set $Y = \frac{1}{\varepsilon}\sum^k_{i=1}\abs{Y_i} \in (L^\Psi)_+$. By Step I,
$$\bigabs{\E[ W_{Y,n}Y_i ]} \leq \E\big[\abs{W_{Y,n}Y_i}\big]  \leq \varepsilon
\E\big[\abs{W_{Y,n}}Y\big]\leq
\varepsilon.$$
This shows that $0$ lies in the $\sigma(H^\Phi,\ls)$-closed convex hull of
$\cA_n$.
Since $\sigma(H^\Phi,\ls)$ is the weak topology on $H^\Phi$, $0$ lies in the
norm closed convex hull of $\cA_n$.

Now take $W_{Y_i,n}\in \cA_n$, $1\leq i\leq k$, and a
convex combination
$W_n = \sum^k_{i=1}c_iW_{Y_i,n}$ such that $\norm{W_n}_\Phi\leq
\frac{1}{2^n}.$
Put $Z_n = \sum^k_{i=1}c_iZ_{Y_i,n}$.
Then 
\[ X_n:=Z_n + W_n = \sum^k_{i=1}c_i(Z_{Y_i,n}+W_{Y_i,n}) \in \C\] by
convexity of
$\C$.
Moreover, since $\Phi$ is a convex function, 
$$ \E[ \Phi(\abs{Z_n})] \leq \sum^k_{i=1}c_i\E\big[\Phi(\abs{Z_{Y_i,n}})\big] 
\leq
\frac{1}{2^n}.$$

\noindent\underline{Step III}. 
In the notation of Step II, a subsequence of $(X_n)$ order converges to $0$. 
Thus $0 \in \ol{\C}^o$.

From Step II ,we know that $\norm{W_n}_\Phi \leq \frac{1}{2^n}$ for all $n$ and
hence $\sum^\infty_1 |W_n| \in L^\Phi$.
Also, since $\Phi$ is continuous and increasing,
\[ \E\Big[\Phi(\sup_n \abs{Z_n}) \Big]= \E\Big[\sup_n\Phi(\abs{Z_n})\Big]\leq
\sum_1^\infty
\E[\Phi(\abs{Z_n})]\leq \sum
\frac{1}{2^n} = 1,\]
from which it follows that $\sup_n \abs{Z_n}\in L^\Phi$.
Therefore,
 $$\widetilde{X}: = \sup_n \abs{Z_n}+
\sum_1^\infty \abs{W_n} \in L^\Phi.$$ 
Obviously, $\abs{X_n}\leq \widetilde{X}$ for all $n\geq 1$.
Thus $(X_n)$ is an order bounded sequence in $L^\Phi$.
By Markov's Inequality, 
$$\Phi(\varepsilon)\mathbb{P}\{\abs{Z_n}>\varepsilon\} \leq
\E[\Phi(\abs{Z_n})]\leq \frac{1}{2^n}$$ for any $\varepsilon >0$.
It follows that $(Z_n)$ converges to $0$ in probability.
Since $\sum_1^\infty
\abs{W_n} \in L^\Phi$, $(W_n)$ converges to $0$ a.s.
Therefore, 
a subsequence of $(X_n)$
converges to $0$ a.s., and thus in order, since the whole sequence $(X_n)$ is
order bounded.
\end{proof}

Theorem~\ref{bdd-C} allows us to characterize  general order
closed convex sets in $L^\Phi$ in terms of the topology $\sigma(L^\Phi,L^\Psi)$.

\begin{cor}\label{orlicz-closed}
Denote by $\mathcal{B}$ the closed unit ball in $L^\Phi$. For a convex set $\C$
in $L^\Phi$, the following statements are equivalent:
\begin{enumerate}
 \item\label{orlicz-closed1} $\C$ is order closed,
\item\label{orlicz-closed2} $\C$ is $\sigma(L^\Phi,L^\Psi)$-sequentially closed.
\item\label{orlicz-closed3} $\C\cap k\mathcal{B}$ is
$\sigma(L^\Phi,L^\Psi)$-closed for all
$k\geq 1$. 
\end{enumerate}
\end{cor}

\begin{proof}
The implication
\eqref{orlicz-closed2}$\implies$\eqref{orlicz-closed1} follows from the
observation at the beginning of this section. 
Theorem~\ref{bdd-C} gives
\eqref{orlicz-closed1}$\implies$\eqref{orlicz-closed3}.
 The implication (3) $\implies$ (2) follows from the fact that every
$\sigma(\lh,\ls)$-convergent sequence is norm bounded.
\end{proof}

Let $\X$ be a Banach space with closed unit ball $\mathcal{B}$ and let $\tau$ be
a locally convex topology on $\X$. Say that $\tau$ has the
\emph{Krein-Smulian property} if a  convex set $\C$ in $\X$ is
$\tau$-closed precisely when $\C\cap k\mathcal{B}$ is
$\tau$-closed for all $k\geq 1$.
The well known Krein-Smulian Theorem says that for any Banach space $\X$, the
weak$^*$ topology $\sigma(\X^*,\X)$ on $\X^*$ has the Krein-Smulian property.
Corollary~\ref{orlicz-closed} leads to the natural question of characterizing
the pairs $(\Phi,\Psi)$ 
so that $\sigma(\lh,\ls)$ has the Krein-Smulian property.
The next lemma is the key construction to solving this question.
A set $\C\subset L^\Phi$ is said to be 
\begin{enumerate}
\item[(i)] 
\emph{monotone} if  $X_1\geq
X_2\in \C$ implies $X_1\in \C$, 
\item[(ii)] \emph{positively homogeneous} if $\lambda \C\subset \C$ for
any $\lambda\geq 0$, and 
\item[(iii)] \emph{additive} if $\C + \C \subset \C$.
 \end{enumerate}
A set that is positively homogeneous and additive is clearly convex.

\begin{lem}\label{main-lemma}
If $\Phi$ and $\Psi$ both fail the $\Delta_2$-condition, then 
$L^\Phi$ admits a monotone, positively homogeneous and additive subset $\C$
which
is order closed but not
$\sigma(L^\Phi,L^\Psi)$-closed. 
Furthermore, for any $X\in L^\Phi$, there exists $k\in\R$ so that $X-k\one
\notin \C$.
\end{lem}

\begin{proof}
Assume that both $\Phi$ and $\Psi$ fail the the $\Delta_2$-condition.
We claim that there are a norm bounded set of disjoint positive random variables
$\{X_n\}_{n\geq 1}\cup \{W_0\} \cup \{W_{ij}\}_{i,j \geq 1}$ in $L^\Phi$ and a
norm bounded set of disjoint positive random variables $\{Y_n\}_{n\geq 1}\cup
\{Z_0\}
\cup
\{Z_{ij}\}_{i,j\geq 1}$ in $L^\Psi$ such that 
\begin{enumerate}
\item[(a)] $\supp Y_n \subset\supp X_n$, $\supp W_0\subset \supp Z_0$ and
$\supp W_{ij}\subset \supp
Z_{ij}  $ for all $n,i,j\geq 1$,
\item[(b)] The pointwise sums $\widetilde{X}:= \sum_n X_n$ and $\widetilde{Z}: =
\sum_{i,j} Z_{ij}$ belong to $\lh$ and $\ls$ respectively,
\item[(c)] $\E[ X_nY_n] = \E[ W_0Z_0]= \E[W_{ij}Z_{ij}]=1$ for
all $n,i,j\geq 1$.
\end{enumerate}
Since $\mathbb{P} $ is nonatomic, there are three disjoint measurable
subsets $\Omega_1,\Omega_2,\Omega_3$ of $\Omega$, each of which is atomless and
has
positive
measure.  Choose any $0\leq W_0\in L^\Phi(\Omega_2)$ and $0\leq Z_0\in
L^\Psi(\Omega_2)$
such that  $\supp W_0\subset \supp Z_0$ and 
$\E[ W_0Z_0]=1$. 
Since  $\Phi$ fails the
$\Delta_2$-condition,
we may apply  \cite[pp.~139, Theorem~5]{RR:91} to $L^\Phi(\Omega_1)$ to  obtain
a sequence $(X_n)$ of normalized
disjoint positive random variables in
$L^\Phi$ such that $\widetilde{X}:= \sum_n X_n \in L^\Phi$.
Choose a norm bounded sequence $(Y_n) \subset \ls$ so that $\supp Y_n \subset
\supp X_n$ and that $\E[X_nY_n] = 1$ for all $n$.
Similarly, since $\Psi$ fails the $\Delta_2$-condition, 
there is a normalized
disjoint positive sequence $(Z_{ij})_{i,j\geq 1} \subset L^\Psi(\Om_3)$  so that
$\widetilde{Z}: =
\sum_{i,j} Z_{ij}\in L^\Psi$.
Then choose $(W_{ij}) \subset L^\Phi$ with the desired properties.

 For any $X\in L^\Phi$, $$ \sum_{i,j}
\bigabs{\E[XZ_{ij}]} \leq \E\big[ \bigabs{X\widetilde{Z}}\big]\leq
\norm{X}_\Phi\bignorm{\widetilde{Z}}_\Psi.$$
Thus the map $\mathcal{T}$ defined by 
$$ \mathcal{T}X = \Big(\E[ XY_n]\Big)_n \oplus \E[XZ_0] \oplus
\Big(\E[XZ_{ij}]\Big)_{ij}$$
is a bounded linear operator from $L^\Phi$ into $\ell^\infty \oplus \R \oplus
\ell^1(\N\times \N)$.
Clearly, $\mathcal{T}$ is a positive operator.
Define the {\em summing operator}  $\mathcal{S}: \ell^1 \to
\ell^\infty$ by $$\mathcal{S}\big((a_j)_j\big) =
\Big(\sum^n_{j=1}a_j\Big)_n.$$
For any $y = (y(i,j))_{i,j\geq 1} \in \ell^1(\N\times\N)$, put $y_i = (y(i,j))_j
\in \ell^1$ for any $i\geq 1$.
Let $\C$ be the subset of $L^\Phi$ consisting of all functions $X\in L^\Phi$ for
which, if we write $\mathcal{T}X= u\oplus a \oplus
v$, there are $\lambda \in \R$ and $ y\in \ell^1(\N\times \N)$ such
that 
\begin{align*}
&\lambda\geq 0, y\geq 0 \text{, and } \sum_i2^i\|y_i\|_1 =1,\\
 &a \geq -\lambda,\ v\geq \lambda y \text{, and } u \geq
\lambda\sum^l_{i=1}4^i\mathcal{S}y_i \text{ for all $l\geq 1$}. \end{align*}
If the above occurs, we write $X\sim (\lambda,y).$

\medskip

\noindent\underline{Claim I}: $\C$ is monotone, positively homogeneous and 
additive.

Indeed, if $X' \geq X\in \C$ and $X\sim (\lambda,y)$,
then clearly $X'\sim (\lambda,y)$ and $\mu X\sim (\mu\lambda,y)$ for any
$\mu\geq 0$, so that $X',\mu X\in
\C$. Now suppose that $X\sim (\lambda,y)$ and $X'\sim (\lambda',y')$.
Since $y,y'\geq 0$, it follows that
\[ \sum_i 2^i\norm{ \lambda y_i + \lambda' y'_i}_1
=\lambda\sum_i 2^i\norm{y_i}_1 +\lambda'\sum_i 2^i\norm{y'_i}_1 =
\lambda+\lambda'.\]
Thus we can find $0\leq y'' \in \ell^1(\N\times \N)$ such that $\sum_i
2^i\|y''_i\|_1
=
1$ and that 
\[(\lambda+\lambda') y'' =   \lambda y +  \lambda' y'.\]
%
%
Let us  show that $X+X'\sim(\lambda+\lambda',y'')$.
Indeed, write $\mathcal{T}X = u\oplus a \oplus v$ and $\mathcal{T}X' = u' \oplus
a' \oplus v'$, then
$$\mathcal{T}(X+X')=(u+u')\oplus (a+a')\oplus (v+v').$$
Now $$a +a' \geq (-\lambda) +
(-\lambda') =
-(\lambda+\lambda'),$$
$$ v+v' \geq \lambda y +
\lambda' y' = (\lambda+\lambda') y'',$$ and
\begin{align*}
u+ u' \geq &\lambda\sum^l_{i=1}4^i\mathcal{S}y_i +
\lambda' \sum^l_{i=1}4^i\mathcal{S}y'_i = \sum^l_{i=1}4^i\mathcal{S}(\lambda y_i
+
\lambda' y_i')\\  
=& (\lambda+\lambda') \sum^l_{i=1}4^i\mathcal{S}y_i'', \quad\text{ for all
$l\geq 1$}.
\end{align*}
This proves that  $X+X'\sim(\lambda+\lambda',y'')$, so that $X +X'
\in \C$, as desired.

\medskip

Since order
intervals in $\ell^1$ are norm compact, for any norm convergent positive
sequence 
$(v_p)$ in $\ell^1$, the set $\cup_p [0,v_p]$ is relatively norm compact in
$\ell^1$.

\medskip

\noindent \underline{Claim II}: $\C$ is order closed in $L^\Phi$.

Let $(U_p)$ be a sequence in $\C$ that order converges to some $U\in L^\Phi$. We
want to show that $U \in
\C$.
Write $$\mathcal{T}U_p = u_p \oplus a_p \oplus
v_p\quad\mbox{and}\quad\mathcal{T}U = u\oplus a \oplus v.$$
For each $n\geq 1$,
denote by $x(n)$ the $n$-th coordinate of a vector $x$ in $\ell^\infty$. By
Dominated Convergence Theorem, for any $n\geq 1$,
$$\lim_p u_p(n)=\lim_p\E[U_pY_n] = \E[ UY_n] =u(n).$$ 
Moreover, since $(U_p)$ is order bounded, and therefore, norm bounded, in
$L^\Phi$, it is easy to see that
$(u_p)$ is norm bounded in $\ell^\infty$.  It follows that $(u_p)$
$\sigma(\ell^\infty,\ell^1)$-converges  to $u$.
Similarly, $(a_p)$ converges to $a$.
For any $(b_{ij}) \in \ell^\infty(\N\times \N)$,
\begin{align*}
\Bigabs{\sum_{i,j}b_{i,j}\E\big[ (U_p-U)Z_{ij}]}
&\leq \E\Big[
\abs{U_p-U}\sum_{i,j}\abs{b_{ij}}Z_{ij} \Big]\\
&  \leq\sup_{i,j}\abs{b_{ij}}\cdot  \E\big[ \abs{U_p-U} \widetilde{Z}\big]  \to
0. \end{align*}
Hence $(v_p)$ converges to $v$ with respect to the topology
$\sigma(\ell^1(\N\times \N),
\ell^\infty(\N\times
\N))$.
Since $\ell^1(\N\times \N)$ has the Schur property, 
$(v_p)$ norm converges to $v$.

For each $p$, suppose that  $U_p\sim (\lambda_p,y_p)$ and 
write $y_{pi} = (y_p(i,j))^\infty_{j=1}$ for each $i$.
Choose $M$ so that $\|u_p\|_\infty\leq M$ for all $p$.
If $l\geq 1$, then
\begin{equation}\label{eq1} M \geq u_p(n) \geq \lambda_p
\sum^l_{i=1}4^i\mathcal{S}y_{pi}(n) \to
\lambda_p
\sum^l_{i=1}4^i\norm{y_{pi}}_1 \text{ as $n\to\infty$}. \end{equation}
In particular, $ M \geq \lambda_p \sum_i 2^i \norm{y_{pi}}_1 = \lambda_p \geq
0$, so that $(\lambda_p)$ is a bounded sequence.
Take a subsequence if necessary to assume that $(\lambda_p)$ converges to some
$\lambda\geq 0$.
If $\lambda =0$, put $y$ to be any positive element in $\ell^1(\N\times \N)$
such
that $\sum_i 2^i\norm{y_i}_1 = 1$, where $y_i = (y(i,j))^\infty_{j=1}$.
Then it is easy to see that $a\geq -\lambda$, $v\geq \lambda y$ and $u \geq
\lambda\sum^l_{i=1}4^i\mathcal{S}y_i$
for all $l$.  Hence, $U\sim(\lambda,y)$, and $U\in \C$.

For the rest of the proof, assume that $\lambda> 0$.
Since $v_p \geq \lambda_py_p\geq 0$ for all $p$ and $(v_p)$ is norm convergent
in $\ell^1(\N\times \N)$, it follows that $(\lambda_py_p)$ is relatively norm
compact in
$\ell^1(\N\times\N)$.
Passing to a subsequence again, we may assume that $(\lambda_py_p)_p$ converges
in norm
to some $z$ in $\ell^1(\N\times \N)$. Set $y = \frac{z}{\lambda}$. 
Then $y\geq 0$, and $(y_p)$ converges to $y$ in norm. 
To complete the proof, we will verify that $U\sim (\lambda,y)$. Clearly, for any
$i\geq 1$, we
have $2^i\norm{y_{pi}}_1\rightarrow 2^i\norm{y_i}_1$ as $p\rightarrow \infty$.
Choose $p_0$ such that $\lambda_p \geq \frac{\lambda}{2}$ for all $p \geq
p_0$.
By $(\ref{eq1})$, if $p \geq p_0$,  then $0 \leq 2^i\norm{y_{pi}}_1 \leq
\frac{M}{\lambda
2^{i-1}}$ for any $i\geq 1$.
It follows from Dominated Convergence Theorem that
\begin{equation}\label{eq2}
\sum_i 2^i\norm{y_i}_1 = \lim_p\sum_i 2^i\norm{y_{pi}}_1= 1.
\end{equation}
Furthermore,
\begin{equation}\label{eq3}
a =\lim a_p \geq -\lim \lambda_p = -\lambda,
\end{equation}
and
\begin{equation}\label{eq4}
v = \lim_p v_p\geq \lim_p \lambda_py_p = \lambda y.\end{equation}
Finally, for each $n$ and each $i$, $\mathcal{S}y_{pi}(n) \to \mathcal{S}y_i(n)$
as $p\to \infty$.
Thus, for any $l$,
$$u(n)=\lim_pu_p(n) \geq \lim_p\lambda_p\sum^l_{i=1}4^i\mathcal{S}y_{pi}(n)=
\lambda\sum^l_{i=1}4^i\mathcal{S}y_i(n)  .$$
This proves that 
\begin{equation}\label{eq5}
 u \geq \lambda\sum^l_{i=1}4^i\mathcal{S}y_i \;\;\text{ for any $l$}.
\end{equation}
Clearly, (\ref{eq2})-(\ref{eq5}) show that $U\sim (\lambda,y)$, as desired.

\medskip

\noindent
\underline{Claim III}:
$-W_0 \in
\overline{\C}^{\sigma(L^\Phi,L^\Psi)}\backslash \C$ and thus  $\C$ is not
$\sigma(L^\Phi,L^\Psi)$-closed.  

Clearly 
\[ \mathcal{T}(-W_0) = 0\oplus -1 \oplus 0.\]
If $-W_0\in \C$, then there would exist $\lambda \geq 0$ and $0 \leq y \in
\ell^1(\N\times
\N)$ such that $-1 \geq -\lambda$, $0 \geq \lambda y$, and $\sum_i 2^i\|y_i\|_1
=1$, where $y_i = (y(i,j))_j$.
It follows that $\lambda \geq 1$, forcing $y =0$, which is impossible. This
proves that $-W_0 \notin \C$.

Next, we show that  $-W_0 \in \overline{\C}^{\sigma(L^\Phi,L^\Psi)}$.
Let  $V_1,\dots, V_l\in L^\Psi$ and $\varepsilon
> 0$ be given.  Set $V =
\frac{1}{\varepsilon}\sum^l_{t=1}\abs{V_t}$. Since $\sup_{ij}\E[W_{ij}V]\leq
(\sup_{ij}\norm{W_{ij}}_\Phi)\norm{V}_\Psi<\infty$, there exists $s\geq 1$ large
enough so that $$\E[W_{ij}V] < 2^{s-1}$$ for all $i,j$.
Since $\sum_n \E[X_nV]  \leq \E[\widetilde{X}V]  < \infty$, 
there exists $r\geq 1$ such that $$\sum^\infty_{n=r}\E[X_nV] 
<\frac{1}{2^{s+1}}.$$
Let $y \in \ell^1(\N\times \N)$ be defined by $y(i,j) = \frac{1}{2^s}$ if $(i,j)
= (s,r)$ and $0$ otherwise.
Simple computations show that if $$X :=
2^s\sum^\infty_{n=r}X_n -W_0 + \frac{W_{sr}}{2^s}$$ then $X\sim (1,y)$, so that 
$X\in \C$.
Moreover, if $1\leq t\leq l$, then
\begin{align*}
\Bigabs{\E[XV_t] - \E[(-W_0)V_t]}
&\leq\varepsilon\E\big[
\abs{X+W_0}V\big]\\&\leq \varepsilon 2^s\sum^\infty_{n=r}\E[X_n V] +
\frac{\varepsilon}{2^s}\E[
W_{sr}V]
 < \varepsilon.
\end{align*}
This proves that $-W_0 \in \overline{\C}^{\sigma(L^\Phi,L^\Psi)}$.

\medskip

Finally, suppose that there exists $X\in \lh$ 
so that $X-k\one \in \C$ for all $k \in\R$.
Let $u$ be the first component of
$\mathcal{T}X$.
Then
the first component of $\mathcal{T}(X-k\one)$ is $u-k(\E[ Y_n])_n$. 
Since $X-k\one \in \C$, $u-k(\E[ Y_n])_n \geq 0$.
Thus 
\[ \big(\E[Y_n]\big)_n\leq \frac{u}{k} \text{ for all $k\geq 1$},\]
from which it follows that 
$\E[Y_n]=0$ for all $n$, contrary to the choice of $Y_n$.
\end{proof}

We are now ready to present the main result of this section.

\begin{thm}\label{orlicz-all}
The following statements are equivalent for an Orlicz space $L^\Phi$ defined on
a nonatomic probability space.
\begin{enumerate}
\item Let $\rho:\lh\to(-\infty,\infty]$ be a coherent risk measure.  Then 
$\rho$ has the Fatou property if and only if there is a set $\cQ$ of nonnegative
random variables in $L^\Psi$, each having expectation $1$, such that
\[ \rho(X) = \sup_{Y\in \cQ}\E[-XY] \text{ for any $X\in L^\Phi$}.\]
\item Let $\rho:\lh\to(-\infty,\infty]$ be a proper convex functional. Define
$\rho^*(Y)=\sup_{X\in L^\Phi}(\E[XY] -\rho(X))$ for any $Y\in L^\Psi$.  Then 
$\rho$ has the Fatou property if and only if 
\[\rho(X)=\sup_{Y\in L^\Psi}(\E[XY]-\rho^*(Y)) \text{ for any $X\in L^\Phi$.}\]
\item Every order closed convex set in $L^\Phi$ is $\sigma(\lh,\ls)$-closed.
 \item Every $\sigma(L^\Phi,L^\Psi)$-sequentially closed
convex set in $L^\Phi$ is $\sigma(L^\Phi,L^\Psi)$-closed.
\item $\sigma(L^\Phi,L^\Psi)$ has the
Krein-Smulian property.
\item Either $\Phi$ or $\Psi$ satisfies the
$\Delta_2$-condition.
\end{enumerate}
\end{thm}

\begin{proof}
(3) $\implies$  (2) $\implies$ (1) follows from Propositions \ref{prop3.1} and
\ref{prop3.2} and the observation that every $\sigma(\lh,\ls)$-closed set is
order closed.
By Corollary \ref{orlicz-closed}, we have (3) $\iff$ (4) $\iff$ (5).
 If $\Phi$ satisfies the
$\Delta_2$-condition, then
$\sigma(L^\Phi,L^\Psi)$ is the weak topology, which has
the Krein-Smulian property. If $\Psi$ satisfies the
$\Delta_2$-condition, then
$L^\Phi$ is the norm dual of $L^\Psi$, and
$\sigma(L^\Phi,L^\Psi)$ is the weak$^*$ topology, which
has the Krein-Smulian property by the Krein-Smulian Theorem.
This shows that (6) $\implies$ (5).

Finally,  suppose
that $\Phi$ and $\Psi$ both fail the $\Delta_2$-condition. Let $\C$ be the set
obtained by applying Lemma \ref{main-lemma}.  Define $\rho: L^\Phi\to
(-\infty,\infty]$ by 
$$\rho(X)=\inf\{m\in
\mathbb{R}:X+m\one\in
\C\}.$$
Using the properties of the set $\C$, it is easy to check that $\rho$ is a
coherent risk measure.
Clearly, $\C\subset \{\rho\leq 0\}$ and  $\{\rho<0\}\subset \C$ by
monotonicity of $\C$.
It follows from the order closedness of  $\C$ that $X\in \C$ if  $\rho(X)=0$.
Therefore, $\{\rho\leq 0\}=\C$, so that $$\{\rho\leq
m\}=\C-m\one\;\;\mbox{ for any }m\in\R.$$ 
Thus $\{\rho\leq m\}$ is order closed for all $m$.  By Proposition
\ref{prop3.1}, $\rho$ has the Fatou property. If condition (1) holds, then from
the representation 
\[ \rho(X) = \sup_{Y\in \cQ}\E[-XY] \text{ for any $X\in L^\Phi$},\]
it is clear that $\C = \{\rho\leq 0\}$ is $\sigma(\lh,\ls)$-closed, contrary to
the choice of $\C$.
This proves that (1) $\implies$ (6).
\end{proof}

\begin{remark}
In the special case where  $\Psi$ satisfies the $\Delta_2$-condition,
Theorem~\ref{bdd-C} was announced by 
Delbaen and Owari at the Vienna
Congress on Mathematical Finance, September 12-14, 2016.
From it they obtained the dual representation result Theorem~\ref{orlicz-all}
(2) {\em when $\Psi$ satisfies the $\Delta_2$-condition}.
Their paper, which appeared in the ArXiv in November 2016, was preceded by the
first version of the present paper (ArXiv October 2016), which already contained
the complete results Theorems \ref{bdd-C} and \ref{orlicz-all}.
\end{remark}


\section{Dual representation with respect to the pair
$(H^\Phi,H^\Psi)$}\label{sec3}

Besides the duality  $(\lh,\ls)$ considered in \S 3, duality theory for risk
measures on $\hh$  with respect to the pair $(\hh,\ls)$ was studied in
\cite{CL:09}, based on the fact that $\ls=(\hh)^*$. We are motivated to study
dual representations of risk measures on $\hh$ using the smaller and more
manageable space $\hs$ as the dual.
Despite Theorem~\ref{orlicz-all} asserts that coherent risk measures on $\lh$
with the Fatou property may fail a
dual representation with respect to the pair $(\lh,\ls)$, the paper
\cite{GX:16a} established a duality theory for risk measures on
$L^\Phi$ with respect to the pair $(\lh,\hs)$ whenever $\lh\neq L^1$.  In this
context, $\hs$
 consists precisely of all random variables $Y$ such that
$$(X_n) \mbox{ is norm bounded in }
\lh,\;X_n\xrightarrow{a.s.}X\quad\Longrightarrow\quad \E[X_nY]\rightarrow
\E[XY].$$
This observation motivated the introduction of a stronger version of the Fatou
property.   A functional $\rho:\hh\rightarrow(-\infty,\infty]$
satisfies the \emph{strong Fatou property} if 
$$(X_n) \mbox{ is norm bounded in }
\hh,\;X_n\xrightarrow{a.s.}X\quad\Longrightarrow\quad\rho(X)\leq \liminf_n
\rho(X_n).$$
In this section, to complement and complete the three cases mentioned above, we
investigate the connection between the strong Fatou property of a coherent risk
measure or a proper convex functional on $H^\Phi$ and its dual representation
with respect to the duality $(\hh,\hs)$.
Say that a set $\C\subset H^\Phi$ is {\em boundedly a.s.\ closed} if for any
norm bounded sequence $(X_n)$ in $\C$ that a.s.\ converges to some $X\in
H^\Phi$,   $X\in \C$.

\begin{prop}\label{uo-tech}
Let $\C$  be a convex set in $\hh$ and let $X \in \hh$.  If
$X\in\overline{\C}^{\sigma(\hh,\hs)}$,
then $X$ is the a.s.-limit of a sequence in $\C$.
The converse holds if $\C$ is norm bounded. Therefore, a norm bounded
convex set $\C$ in $\hh$ is $\sigma(\hh,\hs)$-closed if and only if 
it is boundedly a.s.\ closed.
\end{prop}

\begin{proof}
Let $\C$ be a convex set in $\hh$ and let $X \in \ol{\C}^{\sigma(\hh,\hs)}$. 
Since $\hs$ is a lattice ideal of $\ls=(\hh)^*$, by \cite[Theorem~3.50]{AB:06},
the topological dual of $\hh$ under
$\abs{\sigma}(\hh,\hs)$ is precisely $\hs$.
Thus, by Mazur's Theorem, $$X\in
\overline{\C}^{{\sigma}(\hh,\hs)}=\overline{\C}^{\abs{\sigma}(\hh,\hs)}.$$
Consequently, since $\one\in \hs$, there is a sequence $(X_n) $ in $\C$ such
that 
$$\E[\abs{X_n-X}]\leq \frac{1}{2^n}$$
for all $n\geq 1$.
Since $$\sup_{n\leq m\leq k}(\abs{X_{m}-X}\wedge \one)\big\uparrow_k\,
\sup_{m\geq n}(\abs{X_{m}-X}\wedge \one),$$
it follows that 
\begin{align*}
\E\Big[\sup_{m\geq n}&(\abs{X_{_m}-X}\wedge \one)\Big]=
\lim_k \E\Big[\sup_{n\leq m\leq k}(\abs{X_{_m}-X}\wedge \one)\Big]\\
&\leq\lim_k\E\Big[\sum_{n\leq m\leq k}\abs{X_{_m}-X}\Big]
\leq \frac{1}{2^{n-1}}.
\end{align*}
Therefore, $$\E\Big[\inf_n\sup_{m\geq n}(\abs{X_{_m}-X}\wedge \one)\Big]=0,$$
so that $\inf_n\sup_{m\geq n}(\abs{X_{_m}-X}\wedge \one)=0$. It follows that 
$X_{n}\xrightarrow{a.s.}X$. 

Suppose that $\C$ is a norm bounded convex set  in $H^\Phi$ and that $(X_n)$ is
a sequence in $\C$ that converges a.s.\ to some $X\in \hh$.
Since $L^\Phi=(\hs)^*$ and $\hs$ has order continuous norm,  by
\cite[Theorem~4.18]{AB:06},
for any $Y\in \hs$ and any $\varepsilon>0$,
there exists $X_0\in \lh$ such that
$$\E\big[(\abs{X_n-X}-X_0)^+\abs{Y}\big]<\varepsilon.$$
Therefore, 
\begin{align*}
\bigabs{\E[(X_n-X)Y]}&\leq \E\big[\abs{X_n-X}\abs{Y}\big]
\\
&=
\E\big[(\abs{X_n-X}\wedge
X_0)\abs{Y}\big]+\E\big[(\abs{X_n-X}-X_0)^+\abs{Y}\big]\\
& \leq \E\big[(\abs{X_n-X}\wedge X_0)\abs{Y}\big]+
\varepsilon.
\end{align*}
By Dominated Convergence Theorem, $ \E\big[(\abs{X_n-X}\wedge
X_0)\abs{Y}\big]\rightarrow 0$, and thus
$\limsup_n\bigabs{\E[(X_n-X)Y]}\leq\varepsilon$. Since  $\varepsilon> 0$ is
arbitrary, we conclude that $(X_n)$ $\sigma(\hh,\hs)$-converges to $X$
\end{proof}

The next corollary is the counterpart of Corollary \ref{orlicz-closed} and can
be proved similarly.

\begin{cor}\label{uo-closed-all}Denote by $\mathcal{B}$ the closed unit
ball of $\hh$.
For a convex set $\C$ in $\hh$, the following are
equivalent:
\begin{enumerate}
 \item\label{uo-closed-all1} $\C$ is boundedly a.s.\ closed.
\item\label{uo-closed-all2} $\C$ is $\sigma(\hh,\hs)$-sequentially closed,
\item\label{uo-closed-all3} $\C\cap k\mathcal{B}$ is
$\sigma(\hh,\hs)$-closed for all
$k\geq 1$.
\end{enumerate}
\end{cor}

\begin{thm}\label{orliczh}
The following statements are equivalent for an Orlicz heart $H^\Phi$ defined on
a nonatomic probability space.
\begin{enumerate}
\item Let $\rho:\hh\to(-\infty,\infty]$ be a coherent risk measure.  Then 
$\rho$ has the strong Fatou property if and only if there is a set $\cQ$ of
nonnegative random variables in $H^\Psi$, each having expectation $1$, such that
\[ \rho(X) = \sup_{Y\in \cQ}\E[-XY] \text{ for any $X\in H^\Phi$}.\]
\item Let $\rho:\hh\to(-\infty,\infty]$ be a proper convex functional. Define
$\rho^*(Y)=\sup_{X\in H^\Phi}(\E[XY] -\rho(X))$ for any $Y\in H^\Psi$.  Then 
$\rho$ has the strong Fatou property if and only if 
\[\rho(X)=\sup_{Y\in H^\Psi}(\E[XY]-\rho^*(Y)) \text{ for any $X\in H^\Phi$.}\]
\item Every boundedly a.s.\ closed convex set in $H^\Phi$ is
$\sigma(\hh,\hs)$-closed. 
 \item Every $\sigma(H^\Phi,H^\Psi)$-sequentially closed
convex set in $H^\Phi$ is $\sigma(H^\Phi,H^\Psi)$-closed.
\item $\sigma(H^\Phi,H^\Psi)$ has the
Krein-Smulian property.
\item Either $\Phi$ or $\Psi$ satisfies the
$\Delta_2$-condition.
\end{enumerate}
\end{thm}

\begin{proof}[Sketch of proof]
We omit the proofs of the   implications (5) $\implies$ (4) $\implies$ (3)
$\implies$ (2) $\implies$ (1)  as they are similar to the proofs of the
corresponding
implications in Theorem~\ref{orlicz-all}.

If $\Phi$ satisfies the $\Delta_2$-condition, then $H^\Phi=L^\Phi$ is the dual
space of $H^\Psi$, so
that $\sigma(H^\Phi,H^\Psi)$ is a weak$^*$ topology, which has the
Krein-Smulian property by the Krein-Smulian Theorem.
If $\Psi$ satisfies the $\Delta_2$-condition, then
$H^\Psi=L^\Psi$ is the
norm dual of $H^\Phi$, so that $\sigma(H^\Phi,H^\Psi)$ is a weak topology and
also has the Krein-Smulian property. This proves
(6)$\implies$(5).

Once again, the proof that (1) $\implies$ (6) relies on a construction, whose
verification we postpone for the moment.  If $\Phi$ and $\Psi$ both fail the
$\Delta_2$-condition,  let $\C$ be a subset of $H^\Phi$ as given by Lemma
\ref{lem4.4}.  Then the functional $\rho$ defined on $H^\Phi$ by
\[\rho(X)=\inf\{m\in
\mathbb{R}:X+m\one\in
\C\} \]
is a coherent risk measure on $H^\Phi$ that has the strong Fatou property (due
to the bounded a.s.\ closedness of $\C$) but is not $\sigma(\hh,\hs)$-lower
semicontinuous.  Thus $\rho$ cannot be represented as in condition (1).
\end{proof}

\begin{lem}\label{lem4.4}
Assume that  $\Phi$ and $\Psi$ both fail the $\Delta_2$-condition.
There is a monotone, positively homogeneous and additive subset $\C$
of $H^\Phi$ which is boundedly a.s.\ closed but fails to be
$\sigma(H^\Phi,H^\Psi)$-closed.
\end{lem}

\begin{proof}[Sketch of proof]
Observe that, for any positive random variable $X$ in $\lh$ (respectively,
$\ls$), the truncation $X\one_{\{X\leq k\}}$ lies in $\hh$ (respectively,
$\hs$). Moreover, $\bignorm{X\one_{\{X\leq k\}}}\uparrow_k\, \norm{X}$. As in
the proof of Lemma~\ref{main-lemma}, and applying suitable truncations, we find
a
norm bounded set of
disjoint, positive functions
$\{X_n\}_{n\geq 1}\cup \{W_0\} \cup \{W_{i}\}_{i \geq 1}$ in $H^\Phi$ and a norm
bounded set
of disjoint,  positive functions $\{Y_n\}_{n\geq 1}\cup \{Z_0\} \cup
\{Z_{i}\}_{i\geq 1}$ in $H^\Psi$ such that 
\begin{enumerate}
\item[(a)] $\supp Y_n \subset\supp X_n$, $\supp W_0 \subset\supp Z_0$ and $\supp
W_{i}\subset \supp
Z_{i}$ for all $n, i\geq 1$,
\item[(b)] The pointwise sums $\widetilde{X} = \sum_n X_n \in L^\Phi$ and
$\widetilde{Z} = \sum_i
Z_{i} \in L^\Psi$,
\item[(c)] $\E[X_nY_n]= \E[W_0Z_0] = \E[W_{i}Z_{i}] =1$ for all
$n,i\geq 1$.
\end{enumerate}
Since $Y_n$'s are disjoint, $Y_n\xrightarrow{a.s.}0$, so that $\lim_n\E[XY_n]=0$
for
any $X\in H^\Phi$. Thus, $\big(\E[XY_n]\big)_n \in c_0$.
Also, since $$\sum_i\bigabs{\E[XZ_{i}]} \leq
\E\big[\abs{X}\widetilde{Z}\big]<\infty ,$$ $\big(\E[XZ_{i}]\big)_i
\in  \ell^1$.
Therefore, the map 
\[\mathcal{T}: X \mapsto \Big(\E[XY_n] \Big)_n\oplus \E[XZ_0]  \oplus \Big(\E
[XZ_{i}] \Big)_i\]
 is a bounded positive linear operator from $H^\Phi$ into $c_0 \oplus \R \oplus
 \ell^1$.
For any $y = (y(i,j))_{i,j\geq 1} \in \ell^1(\N\times\N)$, put $y_i = (y(i,j))_j
\in \ell^1$ for any $i\geq 1$.
Let $(s_j)_j$ be the summing basis for $c_0$, i.e., $$s_j = (1,\dots,
1,0,\dots),$$
with $1$ occurring in the first $j$ coordinates, and let $(w_j)$ be
the standard basis for $ \ell^1$, i.e.,
$$w_j=(0,\dots,0,1,0,\dots),$$
with $1$ occurring in the $j$-th coordinate.
Let $\C$ be the subset of $H^\Phi$ consisting of all random variables $X\in
H^\Phi$ for
which, if we write $\mathcal{T}X = u\oplus a \oplus
v$, there are $\lambda \in \R$ and $ y\in \ell^1(\N\times \N)$ such
that 
\begin{align*}
 &\lambda\geq 0,\; y\geq 0, \; \sum_i2^i\|y_i\|_1 =1\mbox{, and }\sum_i
4^i\|y_i\|_1 < \infty.\\
& a \geq -\lambda,\; v\geq \lambda\sum_j\Big(\sum_i4^iy(i,j)\Big)w_j,\;u \geq
\lambda\sum_{j}\Big(\sum_iy(i,j)\Big)s_{j}.
\end{align*}
If the above occurs, we write 
$X\sim (\lambda,y).$

We omit the 
verifications that $\C$ is monotone, positively homogeneous and additive.
 For
$j\geq
1$, define the
projection $\mathcal{P}_j$ on $\ell^1$ by $$\mathcal{P}_j(b_1,b_2,\dots) =
(0,\dots, 0, b_j,
b_{j+1},\dots).$$
Then the condition on $u$ is  equivalent to
\begin{equation}\label{eq6}u \geq
\lambda\Big(\sum_i\bignorm{\mathcal{P}_1y_i}_1,
\sum_i\bignorm{\mathcal{P}_2y_i}_1,\dots\Big).\end{equation}

\noindent\underline{Claim I}: $U\in \C$ whenever there exists a  norm bounded
sequence
$(U_p)_p$ in $\C$ such
that  $(U_p)_p$ converges a.s.~to $U\in H^\Phi$.

Suppose that $$\mathcal{T}U_p = u_p\oplus a_p\oplus v_p,\quad \mbox{and}\quad
\mathcal{T}U = u\oplus a \oplus v.$$
Write $u_p=(u_p(j))_j$ and $u=(u(j))_j$. Then $$u_p(j)=\E[ U_pY_j] 
\to\E[UY_j]=u(j)$$ for each $j$. 
Moreover, since $(U_p)$ is norm bounded in $\hh$, $(u_p)$ is norm
bounded in $c_0$. Thus, $u_p\xrightarrow{\sigma(c_0,\ell^1)}u$. There is a
sequence of  convex combinations, $(\sum_{j=p_{n-1}+1}^{p_n}c_ju_j)$,
$0=p_0<p_1<p_2<\cdots$, which converge to $u$ in the norm of $c_0$. By replacing
$U_p$'s with the corresponding convex combinations (which also converge a.s.~to
$U$), we may assume that $(u_p)$ converges to $u$ in $c_0$-norm.
Similarly,  $v_p\rightarrow v$ coordinatewise and 
\begin{align*}
 \norm{v_p}_1\leq \sum_i\E\big[\abs{U_p}Z_i]= \E\big[\abs{U_p}\widetilde{Z}]\leq
\norm{U_p}_\Phi\norm{\widetilde{Z}}_\Psi.
\end{align*}
Therefore, $v_p\xrightarrow{\sigma(
\ell^1,
c_0)}v.$
Clearly, $a_p\rightarrow a.$

For each $p$, let $U_p \sim (\lambda_p,y_p)$. Write $y_{pi} =
(y_p(i,j))^\infty_{j=1}$ for each
$i$. Then there exists some
$M>0$ such that
\begin{equation}\label{eq7}
 M\geq  \|v_p\|_1 \geq \lambda_p\sum_{i,j}4^iy_p(i,j)=
\lambda_p\sum_i4^i\|y_{pi}\|_1. \end{equation}
In particular, $M\geq \lambda_p \sum_i 2^i \norm{y_{pi}}_1 = \lambda_p
\geq
0$. Thus $(\lambda_p)$ is a bounded sequence.
Passing to a subsequence, we may assume that $(\lambda_p)$ converges to some
$\lambda\geq0$. 
If $\lambda =0$, set $y$ to be any positive element in $\ell^1(\N\times \N)$
such that $\sum_i 2^i\norm{y_i}_1 = 1$ and $\sum_i 4^i\norm{y_i}_1 <\infty$.  It
is easily checked that $U\sim (0,y)$ and hence $U \in \C$.
Assume that $\lambda >0$. 
By $(\ref{eq7})$, for all sufficiently large $p$, 
$$\frac{\lambda}{2}\|\mathcal{P}_j(2^iy_{pi})\|_1\leq{\lambda_p}\|2^iy_{pi}\|_1
\leq
\frac{M}{2^i}$$
for all $i,j\geq 1$.
It follows that, for each $j$,  the sequence
$\big((\norm{\mathcal{P}_j(2^iy_{pi})}_1)_i\big)_{p\geq 1}$, being contained in
an interval of $\ell^1$, is relatively norm compact in
$ \ell^1$.
By passing to a subsequence, we may assume that there exists
$(b_{ij})_i \in  \ell^1$ such that 
\[ \lim_p(\norm{\mathcal{P}_j(2^iy_{pi})}_1)_i = (b_{ij})_i \text{ 
in $ \ell^1$-norm,} \]
for each $j\geq 1$.
In particular,
\[ b_{i,j+1}\leq b_{ij} \text{ and } \sum_ib_{i1} = \lim_p
\sum_i\norm{2^iy_{pi}}_1 = 1.
\]
Set $$y(i,j) = \frac{b_{ij} - b_{i,j+1}}{2^i}\mbox{ for all }i,j\geq 1,$$
and $$y: =
\big(y(i,j)\big) \geq 0.$$
We claim that $U\sim (\lambda,y)$ and hence $U \in \C$. Note that
\[ b_{ij}-b_{i,j+1}=\lim_p(\norm{\mathcal{P}_j(2^iy_{pi})}_1-\norm{\mathcal{P}_{
j+1}(2^iy_{pi})}
_1)=\lim_p2^{i}y_{p}(i,j).\] Thus $y_p(i,j)\rightarrow y(i,j)$ for any
$i,j\geq
1$. 
It follows from Fatou's Lemma and  $(\ref{eq7})$ that
$\sum_i4^i\norm{y_i}_1<\infty.$
Since $(u_p)$ converges to $u$ in $c_0$, $\lim_ju_p(j)=0$ uniformly in $p$.
By condition (\ref{eq6}) for $u_p$, 
$\lim_j\sum_i\norm{\mathcal{P}_jy_{pi}}_1=0$
uniformly
in $p$. In particular, for each $i$, $\lim_j\norm{\mathcal{P}_jy_{pi}}_1=0$
uniformly in
$p$, so
that $\lim_jb_{ij}=0$ for each $i$.
Therefore,  $\|2^iy_i\|_1=\sum_j(b_{ij}-b_{i,j+1}) =b_{i1}$ for all $i$,
and $$ \sum_i \|2^iy_i\|_1 = \sum
b_{i1} = 1.$$
Fatou's Lemma also implies that 
\[ v(j) = \lim_p v_p(j) \geq \liminf_p \lambda_p\sum_i4^iy_p(i,j)\geq
\lambda\sum_i4^iy(i,j),\]
and that (using (\ref{eq6}) for $u_p$)
\[ u(j) = \lim_pu_p(j) \geq \liminf_p \lambda_p\sum_i\norm{P_jy_{pi}}_1 \geq
\lambda \sum_i\norm{P_jy_i}_1.\]
This completes the proof that  $U\sim (\lambda,y)$.

\medskip

\noindent
\underline{Claim II}: $\C$ is not $\sigma(H^\Phi,H^\Psi)$-closed. Precisely,
$-W_0\in \overline{\C}^{\sigma(H^\Phi,H^\Psi)}\backslash \C$.

Since $\mathcal{T}(-W_0) = 0\oplus -1 \oplus 0$, it is easy to see that
$-W_0\notin \C$.
On the other hand, 
observe that,
for
any $s,r\geq 1$, $$X_{sr} := \frac{1}{2^s}\sum^r_{n=1}X_n -W_0 + 2^sW_r\sim
(1,y),$$
where $y(i,j)=\frac{1}{2^s}$ if $(i,j)=(s,r)$ and
$0$ otherwise, so that $X_{sr}\in \C$. As in the proof of Lemma
\ref{main-lemma}, one can show that  $-W_0$
lies in the
$\sigma(H^\Phi,H^\Psi)$-closure
of the double sequence $(X_{sr})_{s,r\geq 1 }$.
\end{proof}


\begin{thebibliography}{100}


\bibitem{AB:06}
C.~Aliprantis, O.~Burkinshaw, \emph{Positive Operators}, Springer Netherlands,
2006.


\bibitem{Arai:10}
T.~Arai, Convex risk measures on Orlicz spaces: inf-convolution and shortfall,
\emph{Mathematics and Financial Economics} 3(2), 2010, 73--88.


\bibitem{Arai:14}
T.~Arai, M.~Fukasawa, Convex risk measures for good deal bounds,
\emph{Mathematical Finance} 24(3), 2014, 464--484.

\bibitem{ADEH:99}
P.~Artzner, F.~Delbaen, J.M.~Eber, D.~Heath, Coherent measures of risk,
\emph{Mathematical Finance} 9, 1999, 203--228.

\bibitem{BF:10} S.~Biagini, M.~Frittelli, On the extension of the Namioka-Klee
theorem and on the Fatou Property for risk measures. In: \emph{Optimality and
risk-modern
trends in mathematical finance}, Springer Berlin Heidelberg, 2010, 1--28.

\bibitem{Biagini:11}
S.~Biagini, M.~Frittelli, M.~Grasselli, Indifference price with general
semimartingales, \emph{Mathematical Finance} 21 (3), 2011, 423--446.

\bibitem{B:11}
H.~Brezis, \emph{Functional Analysis, Sobolev Spaces and Partial Differential
Equations}, Springer Science \& Business Media, 2010.

\bibitem{CL:09}
P.~Cheridito, T.~Li, Risk measures on Orlicz hearts, \emph{Mathematical Finance}
19(2), 2009, 189--214.

\bibitem{DRAPE16}
S.~Drapeau, A.~Jamneshan, M.~Karliczek and M.~Kupper, The
algebra of conditional sets and the concepts of conditional topology and
compactness, \emph{Journal of Mathematical Analysis and Applications} 437(1),
2016, 561--589.


\bibitem{D:02}
F.~Delbaen, Coherent risk measures on general probability
spaces, In: \emph{Advances in finance and stochastics}, Springer Berlin
Heidelberg, 2002, 1--37.




\bibitem{Delbaen:09}
F.~Delbaen, Risk measures for non-integrable random
variables, \emph{Mathematical Finance} 19(2), 2009, 329--333.


\bibitem{DO:16}
F.~Delbaen, K.~Owari, On convex functions on the duals of $\Delta_2$-Orlicz
spaces. Preprint: {\tt arXiv:1611.06218}
	
	
\bibitem{ES:02} 
G.A.~Edgar, L.~Sucheston, \emph{Stopping Times and Directed
Processes}. Encyclopedia of Mathematics and Its Applications, 47, Cambridge:
Cambridge University Press, 1992.


\bibitem{FKM:13}
E.~Farkas, P.~Koch Medina, C.~Munari, Beyond cash-additive risk measures: When
changing the numeraire fails, \emph{Finance and Stochastics} 18(1), 2013,
145-173. 


\bibitem{FILL09}
D.~Filipovic, M.~Kupper, N.~Vogelpoth, {Separation and
duality in
locally $L_0$-convex modules}, \emph{J.\ Func.\ Anal.} 256(12), 2009, 3996-4029.

\bibitem{Fri:02}
M.~Frittelli, E.~Rosazza Gianin, Putting order in risk measures, \emph{Journal
of Banking
and Finance} 26(7), 2002, 1473--1486.

\bibitem{FS:04}
H.~Follmer, A.~Schied, Stochastic Finanace, An introduction in Discrete Time, De
Gruyter Studies in Mathematics 27. 2nd edn. Walter de Gruyter, NY, 2004.

\bibitem{GLMX:17}
N.~Gao, D.~Leung, C.~Munari, F.~Xanthos, Fatou property, representations, and
extensions of law-invariant risk measures on general Orlicz spaces, submitted.
Preprint: {\tt
 arXiv:1610.08806}


\bibitem{GX:16a}
N.~Gao, F.~Xanthos, On the C-property and $w^*$-representations of risk
measures, \emph{Mathematical Finance}, to appear. Preprint: {\tt
arXiv:1511.03159 }

\bibitem{GX:14}
N.~Gao, F.~Xanthos, Unbounded order convergence and application to martingales
without probability, \emph{J.\ Math.\ Anal.\  Appl.} 415(2), 2014, 931-947.

\bibitem{GUO10}
T.X.~Guo, Relations between some basic results derived from two kinds of
topologies for a random locally convex module, \emph{J.\ Func.\ Anal.} 258,
2010, 3024-3047.


\bibitem{HZ:07}
P.~Hajek, V.~Montesinos Santalucia, J.~Vanderwerff, V.~Zizler,
\emph{Biorthogonal Systems in Banach Spaces}, Springer Science \& Business
Media,
2007.

\bibitem{Jouini06}
E.~Jouini, W.~Schachermayer, N. Touzi, Law invariant risk measures have the
Fatou Property. In: \emph{Advances in mathematical economics},
Springer Japan, 2006, 49--71.

\bibitem{KSZ:14}
V.~Kr\"{a}tschmer, A.~Schied, H.~Z\"{a}hle, Comparative and qualitative
robustness for law-invariant risk measures, \emph{Finance and Stochastics}
18, 2014, 271--295.



\bibitem{Orihuela:12}
J.~Orihuela, M.R.~Gal\'an, Lebesgue property for convex
risk measures on Orlicz spaces, \emph{Mathematics and Financial Economics} 6(1),
2012, 15--35.

\bibitem{Owari:14}
K. Owari, Maximum Lebesgue extension of monotone convex functions, \emph{J.\
Func.\ Anal.} 266(6), 2014, 3572-3611.


\bibitem{RR:91}
M.M.~Rao, Z.D.~Ren, \emph{Theory of Orlicz spaces}, Marcel Dekker, Inc., New
York, 1991.




\end{thebibliography}
\end{document}